\documentclass[12pt,onecolumn,journal,draftcls]{IEEEtran}
\usepackage{graphicx}
\usepackage{psfrag}
\usepackage{epsfig}
\usepackage{amsmath}
\usepackage{amssymb}
\usepackage{latexsym}
\usepackage{url}

\newcommand{\be}{\begin{equation}}
\newcommand{\ee}{\end{equation}}
\newcommand{\bea}{\begin{eqnarray}}
\newcommand{\eea}{\end{eqnarray}}
\newcommand{\bes}{\begin{eqnarray*}}
\newcommand{\ees}{\end{eqnarray*}}

\newtheorem {theorem}{Theorem}

\newtheorem {corollary}{Corollary}
\newtheorem {notation}{Remark}

\begin{document}
\title{On the End-to-End Distortion for a Buffered Transmission over Fading Channel}
\author{Qiang Li and C. N. Georghiades
\thanks{The material in this paper was presented
in part at the International Symposium on Information Theory (ISIT),
Seattle, WA, July 2004}. \thanks{The authors are with the Department
of Electrical and Computer Engineering, Texas A\&M University,
College Station, TX, USA. E-mail: \{qiangli,
georghiades\}@ece.tamu.edu.}}\centerfigcaptionstrue \maketitle
\thispagestyle{empty}

\begin{abstract}
In this paper, we study the end-to-end distortion/delay tradeoff for
a analogue source transmitted over a fading channel. The analogue
source is quantized and stored in a buffer until it is transmitted.
There are two extreme cases as far as buffer delay is concerned: no
delay and infinite delay. We observe that there is a significant
power gain by introducing a buffer delay. Our goal is to investigate
the situation between these two extremes. Using recently proposed
\emph{effective capacity} concept, we derive a closed-form formula
for this tradeoff. For SISO case, an asymptotically tight upper
bound for our distortion-delay curve is derived, which approaches to
the infinite delay lower bound as $\mathcal{D}_\infty
\exp(\frac{C}{\tau_n})$, with $\tau_n$ is the normalized delay, $C$
is a constant. For more general MIMO channel, we computed the
distortion SNR exponent -- the exponential decay rate of the
expected distortion in the high SNR regime. Numerical results
demonstrate that introduction of a small amount delay can save
significant transmission power.
\end{abstract}

\newpage

\pagenumbering{arabic}
\section{Introduction}
Quality-of-Service (QoS) is a critical design objective for
next-generation wireless communication system. In general, the data,
voice and multimedia transmission over packet cellular networks,
wireless LAN or sensor networks involves the analogue observations
are transmitted to the end user over a wireless link. End-to-End
distortion and transmission delay are two fundamental QoS metrics.
Such QoS requirements pose a challenge for the system design due to
the unreliability and time varying nature of the wireless link.

\begin{figure}[h]
  \centerline{
  \scalebox{0.7}{\includegraphics{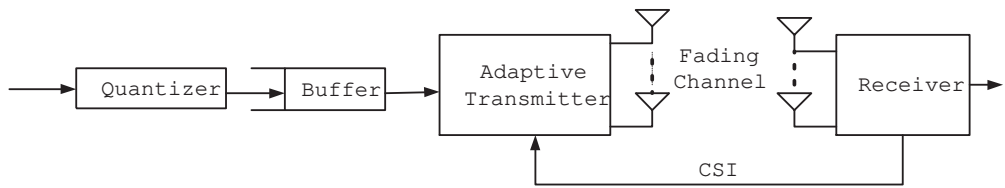}}}
  \caption{System model}\label{fig1}
\end{figure}

 In this paper, we consider transmission of an analogue source over a
wireless time-varying fading channel. Our goal is to optimize the
end-to-end distortion given a delay constraint. We first focus on
the single antenna case (SISO) and derive the distortion and delay
tradeoff for the wireless fading channel. We then extended our model
to multiple input and multiple output (MIMO) block Rayleigh fading
channel. We compute the SNR exponent \cite{CN05} for the buffered
transmission. To this end, we adopt a cross-layer approach shown in
Figure 1. At this point, for simplicity we assume an independent and
identically distributed (i.i.d.) block fading channel model. Such a
model is suitable for serval practical communication scenarios,
e.g., time hopping in TDMA, frequency hopping in FDMA and
multicarrier systems. Extension to more practical time-correlated
case will be discussed later. Throughout this paper, we always
assume channel state information (CSI) is perfectly known at the
receiver and the transmitter only know the instantaneous channel
capacity via a feedback link (transmitter don't need to know the
exact channel realization).

We consider an i.i.d. complex memoryless Gaussian source $\sim
\mathcal{CN}(0,1)$, which is quantized it and then fed into a
buffer. Since the channel is time-varying, the transmitter adjust
the transmission rate to the current channel status. The relevant
performance criteria are the end-to-end quadratic distortion and the
buffer delay. We aim to find the relationship between the distortion
and delay for some average transmission power. The Gaussian source
is a good approximation of more general source distribution in high
resolution regime \cite{LMWA05,GE07}. We assume that each group of
$K$ source samples is tranmsitted over $N$ channel uses on average.
We define the corresponding \emph{bandwidth ratio} as
\begin{equation}\label{eq:bandratio}
  \eta = \frac{N}{K},
\end{equation}
where $K$ is large enough to consider the source as ergodic and $N$
is large enough to design codes that can achieve the instantaneous
channel capacity. Our tools here are the large deviation theory and
information theory.

Recently, some researchers have considered such end-to-end quadratic
distortion as the performance criteria. In \cite{HG05}, Holliday and
Goldsmith first investigated the end-to-end distortion for the MIMO
block fading channel, based on the source-channel separation theorem
and Zheng and Tse's diversity-multiplexing trade-off. And they also
incorporated the delay consideration into their model using ARQ
argument, which is different from our approach. In \cite{LMWA05},
Laneman et al., considered the problem of minimum average distortion
transmission over parallelled channels. They introduced the
distortion SNR exponent as a figure of merit for high SNR value, and
compared the multiple description source coding diversity and
channel coding diversity. Caire and Narayanan \cite{CN05} pointed
out the the separation theorem does not hold for delay constrained
and the unknown channel at the transmitter end, they investigated
the SNR exponent of the distortion function in high SNR regime for
this problem, an upper bound and lower-bound for the distortion SNR
exponent were derived. \cite{GE07} Gunduz and Erkip extended their
results by a layered broadcast transmission scheme. For some
bandwidth ratio, the optimum SNR exponent is achieved.

For the combination of queuing and information theories, in
\cite{WN03}, Wu and Negi, first proposed the concept of effective
capacity, which is an extension of Shannon's capacity by
incorporating into the buffer delay. The effective capacity is the
dual of the Chang's effective bandwidth \cite{CT95} in the network
literature. Negi and Goel \cite{NG04} unified the effective capacity
with error exponent for more practical considerations. A QOS-aware
rate and power control algorithm for wireless fading channel was
proposed by Tang and Zhang \cite{TZ06}.

For buffered transmission, Berry and Gallager investigated the power
and delay tradeoff for communication over fading channel
\cite{BG02}. In \cite{Tse94}, Tse analyzed the distortion for a
fixed line networks, but with adaptive quantizer.

The rest of this paper is organized as follows: in Section II, we
state the problem and show inserting a buffer can save significant
power. We introduce the system model and some preliminaries of the
effective capacity in Section III. Section IV develops our main
results--distortion-delay function and an upper bound for SISO
channel, some asymptotic analysis is provided. In Section V, We
extend the distortion analysis to MIMO channel, and the SNR exponent
for buffered transmission is derived. Distortion-delay for large
antenna MIMO channel is also derived by utilizing the mutual
information Gaussian approximation. Finally, Section VII concludes
the paper.

Throughout this paper, normal letters indicate scalar quantities and
boldface fonts denote matrices and vectors. For any matrix
$\mathbf{M}$ we write its transpose as $\mathbf{M}^{T}$ and
$\mathbf{M}^{H}$ is its conjugate transpose. $x^*$ denotes the
conjugate of $x$. $\ln(\cdot)$ and $\log(\cdot)$ represent the
natural and $2$ based logarithm.

\section{Problem Statement}
 For buffered transmission over the fading channel, there are two extreme cases: 1) There is no
buffer --- no delay, 2) we have an infinite buffer size, i.e., we
allow an infinite transmission delay. For the first case, we
adaptively quantize the Gaussian source according to the CSI.
Assuming perfect transmission, we can approximate the average
achievable quadratic distortion by:
\begin{equation}\label{eq:1}
  \mathcal{D}_{0}(\rho) = \textrm{E}[\exp(-\eta\ln(1+ |h|^2 \frac{P}{N_0W})]~,
\end{equation}
where $P$ denotes the transmission power, $W$ and $N_0$ resent the
bandwidth and noise variance; $h$ is the channel gain, a random
variable with unit variance follow a certain statistical
distribution. Here, we have used the information theoretical
results: Gaussian distortion-rate function can be express as
$\mathcal{D}(R_s)=\exp(-\eta R_c)$ and $C(\rho)=\log(1+|h|^2 \rho)$
is the instantaneous channel capacity-cost function.
 For infinite delay case, the average transmission rate can achieve the ergodic capacity
of a fading channel and the quantizer can simply adopt a constant
output rate. The average distortion is given by:
\begin{equation}\label{eq:2}
  \mathcal{D}_{\infty}(\rho) = \exp(-\eta\,\textrm{E}[\ln(1+ |h|^2 \frac{P}{N_0W})])~.
\end{equation}
The function $\exp(-(\cdot))$ is a covex function. Due to Jensen's
inequality, the distortion $\mathcal{D}_0$ is low bounded by
$\mathcal{D}_\infty$, i.e., $\mathcal{D}_0\geq \mathcal{D}_\infty$.
The two distortion functions are plotted in Figure 2 for a Rayleigh
fading channel. Notice that there is a gap between no-delay and
infinite delay curves. We can call this transmission power gap as
``Jensen's gain''. Note, we assume $\eta = 2$ and a complex Gaussian
source, this is equivalent to a real source with bandwidth ratio of
one.
\begin{figure}
\begin{center}
\includegraphics[width=5.5in,height=3.5in]{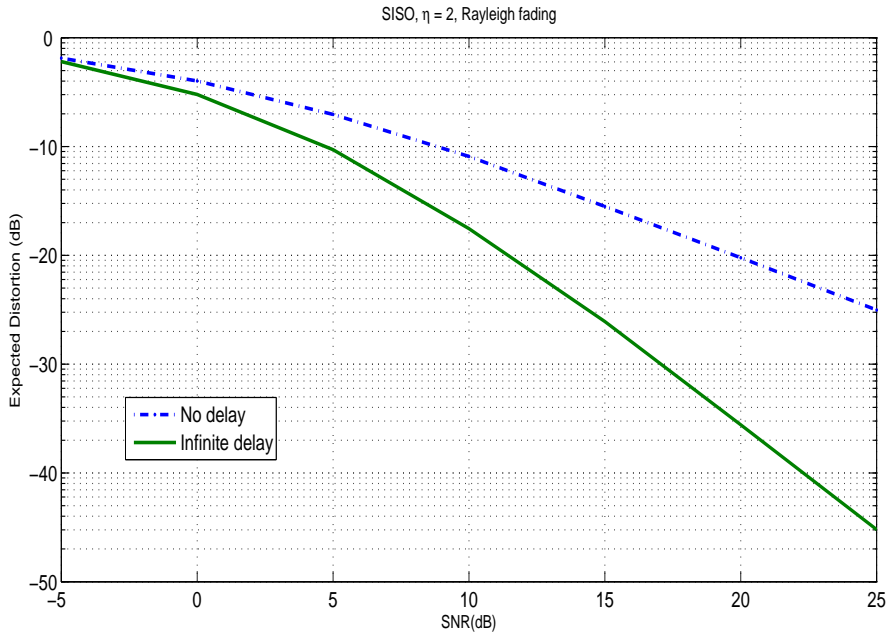}
\caption{Distortion of Gaussian Source Transmitted over i.i.d.
Rayleigh fading Channel.  } \label{Fig:2}
\end{center}
\end{figure}
So introducing a buffer at the transmitter to match the source rate
with the instantaneous quality of the channel can save lots of
transmission power to meet some distortion requirement. Also, we
have simplified the quantization step (constant rate). A natural
question is therefore:
 if we only allow a finite delay or buffer, how much gain can we
achieve? How fast does the distortion curve converge to the
infinite-delay lower bound as the delay increases? One of the the
main result of this paper is a clear characterization the tradeoff
between end-to-end quadratic distortion and delay, which provides
insights to the impact of the buffer delay on the achieved
distortion function of the memoryless analogue source transmitted
over a wireless fading channel.

To answer the question raised earlier, we combine the ideas from the
fields of queuing theory and communication/information theory to
analyze the above problem. The tool we use here is the concept of
\emph{effective capacity} \cite{WN03}, which is the dual of
\emph{effective bandwidth} in networking literature. The effective
capacity synthesizes the channel statistics and QoS metric (delay
and buffer overflow) into a single function using large deviation
theory. It is a powerful and unified approach to study the
statistical QoS performance of wireless transmission where the
service process is time-varying. For i.i.d. SISO block fading
channels we derived a closed-form expression for the
distortion-delay curve, which is hard to analyze due to some
mathematically intractable special functions. Then we give out a
tight upper bound for this distortion-delay function to
theoretically and asymptotically analyze the convergence behavior.

In Fig 2., we find the power gain is marginal for low SNR. As the
SNR value increases, the gain becomes significant. This is because
the $\exp(\cdot)$ and $\log(\cdot)$ functions are approximately
linear in the low SNR regime. Hence, the ``Jensen's gain'' is
negligible at low SNR. We can view the slope of the distortion--SNR
curve as a similarity of the diversity order for the bit error rate
in the wireless communication. Therefore, we will investigate the
distortion SNR exponent for a buffered transmission. Introducing a
buffer can provides some kind of time diversity. For the MIMO
channel, besides the time diversity, we also have space diversity.
We will look into the interplay between these two diversities and
the impact of buffer on the SNR exponent.

\section{System Model}
The system model is illustrated in Figure 1. We have an i.i.d.
complex Gaussian source $\sim \mathcal{CN}(0,1)$ with total
bandwidth $B_w$. We quantize the source samples using vector
quantizer or trellis coded quantizer (TCQ). The quantization operate
every $K$ samples a time and fed into a buffer with size $B$ bits.
Let the $K$ samples have time duration $T_f$, so each frame have
$T_f\times B_w\times R_s = K \cdot R_s$ bits, where $R_s$ bits is
number of bits into which each Gaussian sample is quantized. $K$ is
large enough to ensure ergodic of the source.

We assume a MIMO i.i.d. block fading channel with $M_t$ transmit and
$M_r$ receive antennas. The SISO, MISO and SIMO are special cases of
this general model. The channel model can be expressed as:
\begin{equation}\label{eq:MIMO}
\mathbf{y}_i = \sqrt{\frac{\rho}{M_t}}\mathbf{Hx}_i + \mathbf{w}_i,
\quad \quad i = 1,\cdots,N
\end{equation}
Where $\mathbf{H}$ is the channel matrix containing i.i.d. elements
$h_{i,j}\sim \mathcal{CN}(0,1)$ (Rayleigh independent fading).
$\mathbf{x}_i$ is the transmitted signal at time $i$, the codeword
$\mathbf{X} = [\mathbf{x}_1,\cdots,\mathbf{x}_N] \in
\mathbb{C}^{M_t\times N}$ is normalized so that is satisfies
$\text{tr}(\mathbb{E}[\mathbf{X}^H\mathbf{X}])\leq M_tN$. $\rho$
denotes the signal-to-noise ratio (SNR), defined as the ratio of the
average received signal energy per receiving antenna to the noise
per-component variance. $\mathbf{Z} =
[\mathbf{z}_1,\cdots,\mathbf{z}_N] \in \mathbb{C}^{M_r\times N}$ is
the complex additive Gaussian noise with i.i.d. entries
$\mathcal{CN}(0,1)$. We define $M_*= min(M_t,M_r)$ and $M^* =
max(M_t,M_r)$.


\subsection{Effective Capacity}
The key idea of effective capacity is that, for a dynamic queuing
system with stationary ergodic arrival and service process, the
queue length $Q(t)$ converges in distribution to a random variable
$Q(\infty)$. The probability of queue length exceeding a certain
threshold $B$ decays exponentially fast as the threshold $B$
increases \cite{WN03}. Mathematically,
\begin{equation}\label{eq:3}
\lim_{B \rightarrow \infty}\frac{-1}{B}\ln Pr\{Q(\infty)>B
\}=\theta~,
\end{equation}
where $\theta$ is the QoS parameter decided by the delay requirement
of the queue system. A large value of $\theta$ leads to a stringent
delay requirement, i.e., small delay. In particular, as $\theta$
goes to $\infty$, the system can not tolerate any delay. On the
other end, when $\theta$ goes to $0$, the system can tolerate an
arbitrarily delay.

Let the sequence $\{R[i], i=1,2,\ldots\}$ denote the discrete-time
instantaneous channel capacity, which is a stationary and ergodic
stochastic process. Define
\begin{equation}\label{eq:4}
S[t] \triangleq \sum_{i=1}^{t}R[i]
\end{equation}
as the accumulate service provided by the channel. Assume the
G$\ddot{a}$rtner-Ellis limit of $S[t]$:
\begin{equation}\label{eq:5}
\Lambda_C(\theta) \triangleq
\lim_{t\rightarrow\infty}\frac{1}{t}\ln\text{E}\Big\{e^{\theta
S[t]}\Big\}, \quad\forall\,\theta >0
\end{equation}
exits and is a convex function differentiable for all real $\theta$.
Then, the effective capacity with delay constraint decided by
$\theta$ is defined as
\begin{equation}\label{eq:6}
E_C(\theta) \triangleq -\frac{\Lambda_C(-\theta)}{\theta}
=-\lim_{t\rightarrow\infty}\frac{1}{\theta t}\ln
\text{E}\Big\{e^{-\theta S[t]}\Big\}~.
\end{equation}
In particular, for i.i.d. cases, the effective capacity simply
reduces to the ratio of log-moment generating function of the
instantaneous channel capacity to the exponent $\theta$
\begin{equation}\label{eq:eff_cap}
E_C(\theta) = -\frac{1}{\theta}\ln\text{E}\Big\{e^{-\theta
R[t]}\Big\}~.
\end{equation}
The effective capacity falls into the large deviation framework,
which is asymptotically valid for a large queue size.
\section{Distortion-Delay Function} We will derive the closed-form
expression for the end-to-end quadratic distortion given the delay
constraint in this section. The starting point is vector
quantization and delay bound violation probability using effective
capacity. For a Gaussian source vector $\mathbf{u}$ with  $K$
samples that has support on $\mathbb{C}^K$, a $KR_s$-nats quantizer
is applied to $\mathbf{u}$ via a mapping $\mathbf{u}\rightarrow
\tilde{\mathbf{u}}$. The cardinality of discrete set
$\tilde{\mathbf{u}}$ is $e^{KR_s}$. Define the average quadratic
distortion by
\begin{equation}\label{eq:8}
\mathcal{D}_Q(R_s) \triangleq \frac{1}{K}
\text{E}[|\mathbf{u}-\tilde{\mathbf{u}}|^2]~,
\end{equation}
where the expectation is with respect to $\mathbf{u}$. According to
the distortion-rate theory, the distortion function
$D_Q(R_s)=\exp(-R_s)$ is achievable for a complex Guassian source.
The quantized bits are transmitted over a statistical channel, let
$P_e$ denote the error probability of this channel. It has been
shown in \cite{HZ97} that the achievable end-to-end distortion for
such tandem scheme is upper bounded by
\begin{equation}\label{eq:9}
\mathcal{D}_{e-e}(R_s) \leq \mathcal{D}_Q(R_s) + O(1)P_e ~.
\end{equation}
For our problem, if we assume using Gaussian code to achieve the
instantaneous capacity, the delay bound violation (buffer overflow)
probability will dominate the decoding error probability. From the
effective capacity theory, we have the following approximation for
$P_e$:
\begin{equation}\label{eq:10}
P_e \triangleq P_r\{Q(\infty)\geq B\}\thickapprox \kappa e^{-\theta
B}~,
\end{equation}
where $\theta$ is the QoS parameter, B is the buffer size; $\kappa$
is a constant that denotes the probability that the buffer is
non-empty. $\kappa$ is large compared with $P_e$. Given the delay
constraint at $\tau$ seconds, using Little's theorem, we have
following result: $B = Rs\times B_w \times \tau$. $B_w$ is the
source bandwidth. Substitute (\ref{eq:10}), $B$ and $D_Q(R_s)$ into
(\ref{eq:9}), we may write the bound on the end-to-end distortion as
\begin{equation}\label{eq:ee_dist}
\mathcal{D}_{e-e}(R_s) \leq \exp(-R_s) + O(1)\kappa \exp(-\theta B_w
R_s \tau) ~.
\end{equation}
In order to get analytical results, we consider the asymptotically
large delay and high SNR regime, i.e., small distortion. We can
optimize the end-to-end distortion by choosing the two exponents
equal to each other (exponential order tight). As a result, we have
$\theta = \frac{1}{B_w \tau}$.

If we assume the transmitter don't know the channel realization, but
know the value of instantaneous capacity via the feedback link. The
instantaneous capacity can be achieved by the Gaussian codebook. We
have following theorem.
\begin{theorem}\label{th:th1}
Given a delay $\tau = \frac{1}{B_w \theta }$ and bandwidth raio
$\eta$, the distortion upper bound function of the i.i.d MIMO block
fading channel can be expressed as:
\begin{align}\label{eq:th1_1}
\mathcal{D}(\theta) \leq
\Big[\mathbf{B}^{-1}\det[\mathbf{G}(\theta)]\Big]^{\frac{1}{K\theta}}~.
\end{align}
where $\mathbf{B}= \prod^{M_*}_{i=1}\Gamma(d+i)$, and $d=M^*-M_*$.
And $\mathbf{G}$ is $M_*\times M_*$ Hankel matrix whose $(i,j)$th
entry is defined to be
\begin{equation}\label{eq:th1_2}
g_{i,j} = \int^{\infty}_0
\Big(1+\frac{\rho}{M_t}\lambda\Big)^{-\theta K
\eta}\lambda^{i+j+d}e^{-\lambda}d\lambda, \quad i,j =
0,\cdots,M_*-1~.
\end{equation}
$\Gamma$ is the complete Gamma function.
\end{theorem}
\begin{proof}
The Mutual information for the each MIMO block transmission can be
expressed as:
\begin{equation}\label{eq:th1_3}
R_s(\mathbf{H}) = K\eta \cdot
\ln\det\Big(\mathbf{I}+\frac{\rho}{M_t}\mathbf{HH}^H\Big)
\end{equation}
plug into equation (\ref{eq:eff_cap}) and (\ref{eq:ee_dist}), we
have
\begin{align}\label{eq:th1_4}
\mathcal{D}(\theta) &\leq \bigg\{\textrm{E}\bigg[
\det\Big(\mathbf{I}+\frac{\rho}{M_t}\mathbf{HH}^H\Big)\bigg]^{-\theta
K \eta}\bigg\}^{\frac{1}{\theta K}} \notag\\
&=\bigg\{\int^{\infty}_{0}\prod
\Big(1+\frac{\rho}{M_t}\lambda_i\Big)^{-\theta K
\eta}f(\boldsymbol{\lambda})d\boldsymbol{\lambda}\bigg\}^{\frac{1}{\theta
K}}~.
\end{align}
Where $0\leq \lambda_1 \leq \cdots \leq \lambda_{M_*}$ denote the
ordered eigenvalues of $\mathbf{HH}^H$. The joint distribution of
the $\lambda_i$'s follows the Wishart pdf given by
\begin{equation}\label{eq:th1_5}
f(\boldsymbol{\lambda})
=K^{-1}_{M_t,M_r}\prod^{M_*}_{i=1}\lambda_i^{M^*-M_*}\prod_{i<j}(\lambda_i-\lambda_j)^2\exp\Big(-\sum_i
\lambda_i \Big)~,
\end{equation}
where $K_{M_t,M_r}$ is a normalization constant. Follow the results
of \cite{WG04}, we can get the distortion function as
(\ref{eq:th1_1}).
\end{proof}

\noindent \textbf{Remarks}
\begin{itemize}
  \item If we assume the quantization process is independent of the channel
  status, we can show the the constant quantization rate is the optimum
  one. First, for a buffered system with independent arrival and departure processes, the constant
  arrival processe is optimal with respect to the buffer overflow probability, for all the arrival processes that have
  the same average rate \cite{CT95}. Second, given a buffer overflow probability, constant
  rate quantization will minimize the distortion according to the
  Jensen's inequality. Therefore, constant rate quantization is optimal if
  the quantization process is independent of the channel mutual information.
  Another advantage of constant rate quantization is to reduce the
  quantizer design complexity.
  \item When the quantizer rate selection is according to the buffer
  state status. We can not prove the constant rate quantization is optimal. Hence
  the distortion of (\ref{eq:th1_1}) is an upper bound. One extreme case is that the
  quantizer is chosen to make sure no buffer
  overflow, i. e. , the quantization rate selection is to match the channel
  mutual information profile. This scheme will degenerate to no buffer
  (delay) case. Therefore, it is serious suboptimal. The optimal
  quantizer rate should balance the ``Jensen's gain'' and the
  reduced distortion by decreasing the buffer overflow probability
  via quantization rate matching the buffer status.
\end{itemize}
The introducing buffer delay in (\ref{eq:th1_4}) can be viewed as
first shrinks the integrand near to $1$ as $\theta \rightarrow 0$,
and then restore it after taking the expectation. From Fig. 3, we
can observe that after the contraction function of
$(\cdot)^{\theta}$, as $\theta$ goes to zero, the integrand function
become more linear. This observation can explain why we have a large
gain after introducing a buffer delay mathematically, and provide
some intuitions of distortion--delay function. Moreover, Fig. 3
shows that the large the bandwidth ration $\eta$, the more effective
of the shrink operation (larger gain). Therefore, introduce a buffer
delay has larger gain for high bandwidth ratio scenario, or high
resolution quantization. We will confirm the result later
theoretically by deriving the SNR exponent.
\begin{figure}
\begin{center}
\includegraphics[width=5.5in,height=3.5in]{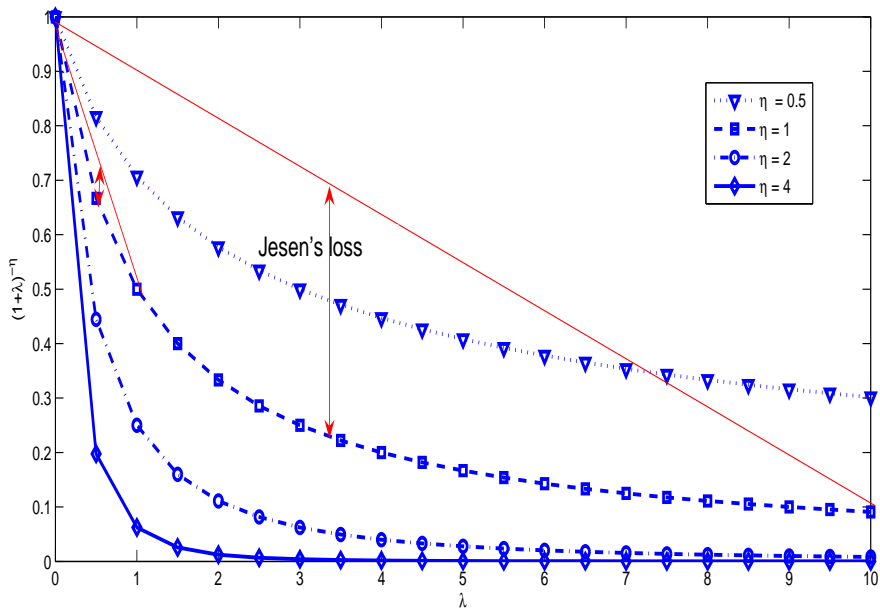}
\caption{Illustration of buffer delay effect on the distortion }
\label{Fig:3}
\end{center}
\end{figure}

The result of Theorem 1. is very complicate, not so much insight can
be given from the expression itself. In the ensuing part of this
paper, we will first investigate the distortion-delay of SISO, MISO
/ SIMO case, which a simpler form can be arrived. Then, for more
general MIMO channel, we consider the high SNR regime and compute
the distortion SNR exponent. Guassian approximation of MIMO mutual
information will also be used to derive an approximation for large
the antenna system.

\subsection{Single Antenna System (SISO)}
For simplicity, we introduce the normalized delay as
$\tau_n=\tau/T_f=\frac{1}{\theta B_w T_f} = \frac{1}{K \theta}$. For
the SISO Rayleigh fading channel, the channel matrix degenerate to a
scalar channel. We have following Corollary.
\begin{corollary}\label{cor:1}
For SISO system, the distortion-delay upper bound is
\begin{equation}\label{eq:cor1_1}
\mathcal{D}(\lambda\eta)\leq \bigg[\rho
^{-\lambda\eta}\exp\Big(\frac{1}{\rho}\Big)\gamma\Big(1-\lambda\eta,\frac{1}{\rho}\Big)\bigg]^{\frac{1}{\lambda}}~,
\end{equation}
where $\lambda = \frac{1}{\tau_n}$ and $\gamma(\cdot,\cdot)$ is the
incomplete Gamma function.
\end{corollary}
\begin{proof}
For SISO channel, the (\ref{eq:th1_1}) is reduced to the scaler
case,
\begin{align}\label{eq:cor1_2}
\mathcal{D}(\lambda)\leq \bigg[\int^\infty_0\Big(1+\rho
x\Big)^{-\lambda\eta}e^{-x}dx\bigg]^{\frac{l}{\lambda}}~,
\end{align}
by the formula of \cite{GR92}, we can complete the proof.
\end{proof}
The closed-form expressions of (\ref{eq:cor1_1}) is very difficult
to analyze due to the special functions. In order to analyze
distortion as the delay constraint increases, it is desirable to
reduce the function into some simple form that is easy to handle.
This objective motivates us to derive an asymptotically tight upper
bound for the distortion-delay function in next section.
\subsubsection{Asymptotic Analysis}
We start by characterizing the behavior of the tail of
distortion-delay curve $\mathcal{D}(\tau_n)$, hence we are
interested in the asymptotically large delay regime. We will only
consider Rayleigh fading SISO case. In this part, we assume $\eta =
1$ for simplicity, generalizing to other bandwidth ratio is
straightforward. We try to show that $\mathcal{D}(\tau_n)\rightarrow
\mathcal{D}(\infty)$ as $\tau_n\rightarrow\infty$. In addition, we
will prove that the limit is approached as $e^{\frac{C}{\tau_n}}$ by
finding the upper bound on the distortion-delay function and then
show the bound is asymptotically tight. The ergodic capacity of
$m_{th}$-order diversity Raleigh fading channel with a constant
transmission power can be expressed as \cite{CTB99}:
\begin{align}\label{eq:16}
C_{erg} &= \frac{\gamma(m,-m/\rho)}{\Gamma(m)}E_1(m/\rho)
+\sum^{m-1}_{k=1}\frac{1}{k}\frac{\gamma(k,m/\rho)\gamma(m-k,-m/\rho)}{\Gamma(k)\Gamma(m-k)}~,
\end{align}
where $\gamma(\cdot,\cdot)$ and $\Gamma(\cdot)$ denote incomplete
and complete Gamma functions; $E_1(\cdot)$ presents the exponential
integration function. Hence for $m=1$, the lower bound of
distortion/delay function can be written as:
\begin{equation}\label{eq:17}
\mathcal{D}(\infty) = \exp\Big(-
e^{\frac{1}{\rho}}E_1(1/\rho)\Big)~.
\end{equation}
Next, We try to derive the asymptotic upper bound on
$\mathcal{D}(\tau_n)$ of (\ref{eq:cor1_1}) to achieve the limit
$\mathcal{D}(\infty)$. We mean asymptotically in the sense of
$\tau_n\rightarrow\infty$ or $\lambda\rightarrow 0$.
\begin{theorem}\label{th:2}
An asymptotic upper bound for $\mathcal{D}(D_n)$ can be expressed
as:
\begin{equation}\label{eq:18}
\mathcal{D}_{upper}(\lambda) =\bigg[
\frac{1}{\lambda-1}\big(e^{\frac{1}{\rho}}-1\big)+
\frac{1}{1-\xi\lambda+\phi\lambda^2}{\rho}^{-\lambda}e^{\frac{1}{\rho}}\bigg]
^{\frac{1}{\lambda}}~,
\end{equation}
where $\xi=0.577215$ and $\phi = \frac{1}{12}(6\xi^2-\pi^2)$. As
$\lambda \rightarrow 0$ this upper bound is asymptotically tight and
approaches $\mathcal{D}(\infty)$ as $\mathcal{D}(\infty)\cdot e^{C
\lambda}$, where $C$ is some constant.
\end{theorem}
\begin{proof}
See Appendix B.
\end{proof}

\subsubsection{Example 1.} We present some numerical results
to verify our findings. Suppose we have a real Gaussian source
$~N(0,1)$ with bandwidth $100kHz$, bandwidth ratio $\eta =
1$\footnote{A real Gaussian source is equivalent  to a complex one
with doubled bandwidth ratio}. We assume an i.i.d. block Rayleigh
fading channel model. Let the duration of each time frame be $2ms$
such that each data frame consists of $200$ source samples. Fig. 4
shows a normalized delay of $5T_f$ can achieve most of the gains,
especially for high transmission power. The gap between this curve
and the infinite delay case is less than $1$dB for typical SNR
value.
\begin{figure}
\begin{center}
\includegraphics[width=5.5in,height=3.5in]{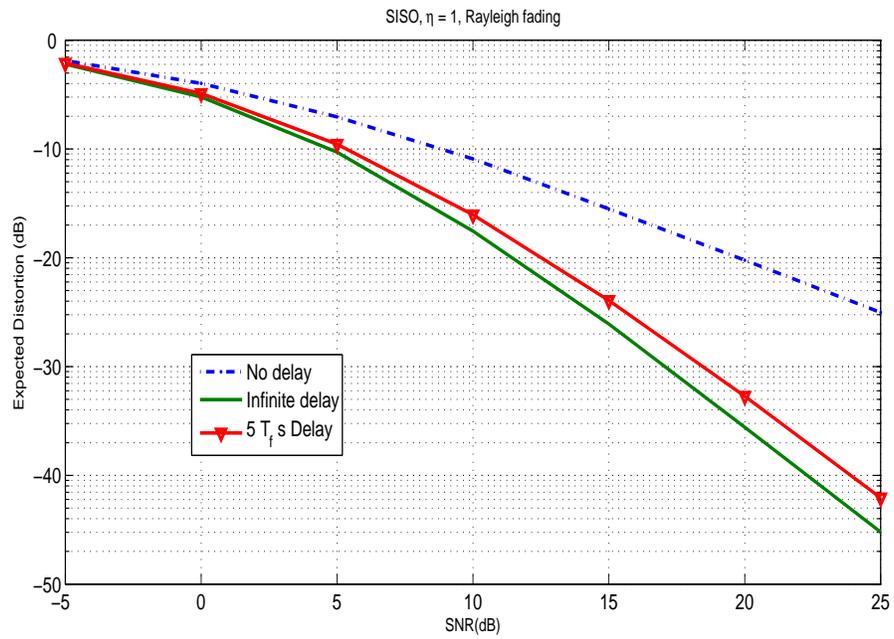}
\caption{Distortion of Real Gaussian Source Transmitted over i.i.d.
Rayleigh fading Channel. } \label{Fig:4}
\end{center}
\end{figure}
In Fig. 5, we plot the end-to-end quadratic distortion vs. SNR and
delay. It clearly characterizes the distortion and delay tradeoff
for the Gaussian source transmitted over the wireless fading
channel. Note that the higher the SNR value, the faster the
distortion converges to the infinite delay lower bound. For SNR
value of $25$dB, less than $2T_f$ delay can achieve most of the
Jensen's gain.
\begin{figure}
\begin{center}
\includegraphics[width=5.5in,height=3.5in]{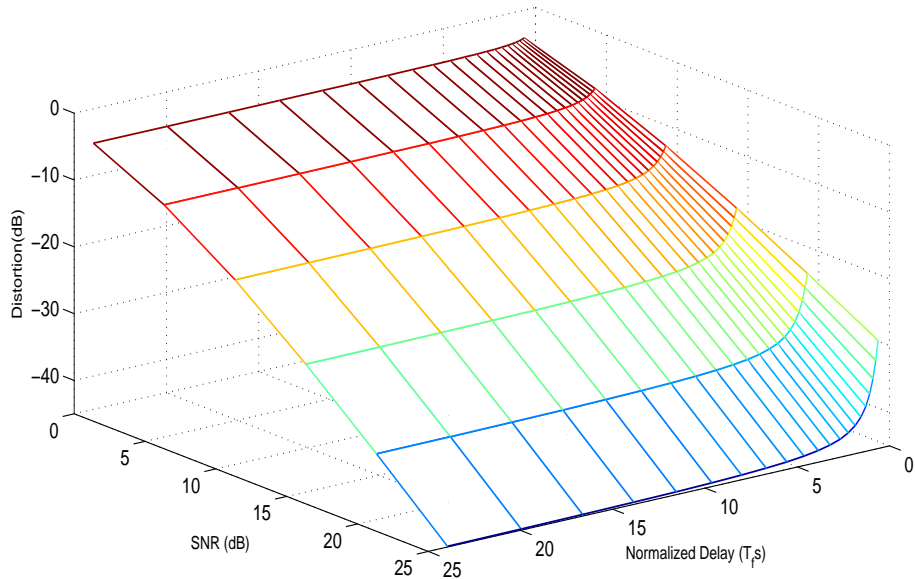}
\caption{Distortion vs. Delay and SNR } \label{Fig:5}
\end{center}
\end{figure}

Fig. 6 shows the upper bound for the distortion/delay
$\mathcal{D}(D_n)$ curve at SNR $=15$dB. The ergodic Shannon
capacity in this case is $3.0015$ nats/symbol and the distortion
$\mathcal{D}(\infty)$ is $0.0025$. The rate of distortion/delay
curve and the upper bound converge to the infinite delay lower bound
is clearly illustrated in Figure 5. It shows the upper bound is
asymptotically tight and converges. From this upper bound and the
distortion/delay function, we observe that introducing some finite
delay can help achieving the $\mathcal{D}(\infty)$ lower bound very
fast. In some practical applications, e.g., video transmission over
wireless fading channel, which can tolerate certain amount of delay,
our results suggest that inserting a buffer between quantizer and
transmitter will enhance the image quality significantly.
Intuitively, a transmission delay can be thought of as some delay
diversity corresponding to space diversity in MIMO channel. Hence
there is also some diversity-rate tradeoff for our problem, which
can lead to results similar to those in \cite{CN05}.
\begin{figure}
\begin{center}
\includegraphics[width=5.5in,height=3.5in]{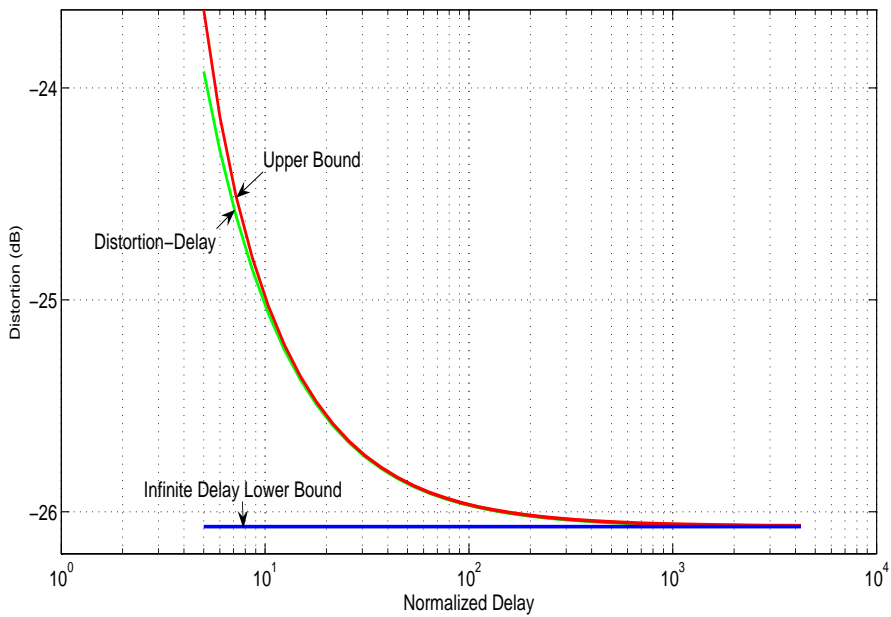}
\caption{Upper Bound of Distortion/Delay function (SNR=15dB) }
\label{Fig:6}
\end{center}
\end{figure}

\subsection{SIMO/MISO Antennas System}
For a SIMO channel of $m$ receiver antenna. We can consider such
channel as a $m_{th}$-order combining diversity Rayleigh fading
channel. Again we here assume $\eta = 1$ for simplicity. The channel
gain after combining is Chi-square distributed with $2m$ degrees of
freedom, and the probability density function (pdf) is given by:
\begin{equation}\label{eq:12}
f(x)= \frac{1}{(m-1)!}x^{m-1}e^{-x}, \quad x>0~.
\end{equation}
\begin{corollary}\label{th:1}
For the SIMO Rayleigh fading channel with m receive antennas. The
distortion-delay upper bound has a closed-form expression:
\begin{align}\label{eq:14}
\mathcal{D}_m(\tau_n)\leq \bigg [\frac{\Gamma(\lambda
-m)}{\Gamma(\lambda)}\rho^{-m}
 \,_1F_1\Big(m;m-\lambda+1;\frac{1}{\rho}\Big)
+\frac{\Gamma(m-\lambda)}{\Gamma(m)}\rho^{-\lambda} \,
_1F_1\Big(\lambda;\lambda-m+1;\frac{1}{\rho}\Big)\bigg]^{\tau_n} ~,
\end{align}
where $\lambda = 1/\tau_n $.
\end{corollary}
\begin{proof}
We start from Eqn. (\ref{eq:th1_1}), with SIMO case
\begin{align}\label{aeq:1}
\mathcal{D}(\theta)&=
\Big(\int^{\infty}_{0}(1+\rho x)^{-\lambda}f(x)d\rho\Big)^{\tau_n}\notag\\
 &= \bigg( \frac{1}{(m-1)!}
\cdot \int^{\infty}_{0}(1+\rho x)^{-\lambda}x^{m-1}e^{-x}dx
\bigg)^{\tau_n}~,
\end{align}
where we have used the expression of $f(x)$ in (\ref{eq:12}). We
know that \cite[Ch. 3.383.5]{GR92}:
\begin{equation}\label{aeq:2}
\int^{\infty}_{0}e^{-px}x^{q-1}(1+ax)^{-v}dx=a^{-q}\Gamma(q)\Psi\big(q,q+1-v;\frac{p}{a}\big)~,
\end{equation}
where $\Psi(\cdot,\cdot;\cdot)$ denotes the degenerate
Hypergeometric function. Reducing to the more commonly used
confluent hypergeometric function, we have following relation:
\begin{align}\label{aeq:3}
\Psi(x,y;z) = &\frac{\Gamma(1-y)}{\Gamma(x-y+z)} _1F_1(x;y;z) +
\frac{\Gamma(y-1)}{\Gamma(x)}z^{1-y} _1F_1(x-y+1;2-y;z)~.
\end{align}
Let $p = 1,\, q=m, \, v = \lambda$ and $a=\rho$. Plugging
(\ref{aeq:3}) into (\ref{aeq:2}), we can prove Lemma 1.

\end{proof}

For MISO case\footnote{We assume transmitter has CSI for MISO case
for beamforming transmission}, it is similar to the SIMO case but
dividing the power by $m$. Even for the SIMO/MISO case the
distortion-delay upper bound function is very complicate. We can
only get some numerical results. Therefore, for more general MIMO
channel, we resort to the SNR exponent in high SNR regime to
demonstrate the buffer gain.

\section{Distortion Exponent of MIMO Block fading channel}
For MIMO block fading channel with a buffered transmission, Eqn.
(\ref{eq:th1_1}) is very hard to analyze and provides less insight.
We can only use the numerical method to compute the function. Since
the ``Jensen's gain'' is negligible in low SNR regime and become
significant at high SNR. Therefore we are more interested in the
high SNR behavior of the expected distortion. We defined the figure
of merit of \emph{distortion exponent} \cite{CN05} with bandwidth
ratio $\eta$:
\begin{equation}\label{eq:exp_1}
\alpha(\eta) = -\lim _{\rho\rightarrow
\inf}\frac{\log\mathcal{D}(\rho,\eta)}{\log\rho}~.
\end{equation}
A distortion exponent of $\alpha$ means that the expected distortion
decays as $\rho^{-\alpha}$ with increasing SNR value $\rho$ when the
SNR is high. We want to characterize the buffer delay and bandwidth
ratio's impact on the SNR exponent.

\begin{theorem}\cite{CN05} (No Buffer)
For transmission of memoryless, complex Gaussian source over a MIMO
block fading channel, the distortion exponent with perfect known
channel is given by
\begin{equation}\label{eq:exp_2}
\alpha(\eta) = \sum_{i=1}^{M_*}\min\Big(\eta,2i-1+|M_t-M_r|\Big).
\end{equation}
\end{theorem}
The proof of Theorem 3, using the technique of \cite{ZT03}.
Intuitively, when the bandwidth ratio is low, the distortion is
limited by the $\eta$ and the degree of freedom of MIMO channel --
the total degree freedom utilized to transmit the information. One
the other hand, when the bandwidth ratio is high, we need more
diversity to provide the transmission reliability. Hence, for high
bandwidth ratio, the system is diversity limited and the SNR
exponent is determined by the second term.

\begin{theorem}(with buffer delay)
For transmission of memoryless, complex Gaussian source over a MIMO
block fading channel, If the quantized bits are stored in a buffer
before transmitting over the fading channel. Assume the transmitter
know exactly the instantaneous channel capacity, the distortion SNR
exponent is given by
\begin{equation}\label{eq:exp_3}
\alpha(\eta) =\tau_n \min
\Big\{\frac{\eta}{\tau_n},2i-1+|M_r-M_t|\Big\}~.
\end{equation}
\end{theorem}
\begin{proof}
Proof can be found in Appendix II.
\end{proof}
\noindent \textbf{Remarks}
\begin{itemize}
  \item We found the SNR exponent of Theorem 4 is similar as the one
  of joint encoding and decoding of $L$ MIMO fading blocks. However,
  the joint encoding increase the transmitter and receiver
  complexity. Introduce a simple buffer delay can get the same SNR
  exponent by utilizing the time diversity.
  \item For SIMO/MIMO case, the SNR exponent reduces to
  $\min\{\eta,\tau_n M\}$, where $M$ is the receiver / transmitter
  antenna number. We can consider $\eta = \tau_n M$ as a corner point.
  Below this point, the system is degree of freedom limited, hence
  introduce more antenna will not improve the SNR exponent. Beyond
  this point, the system is diversity limited. Increasing the
  antenna number to provide more combining branches  that will increase diversity, hence SNR exponent is
  also increased.
\end{itemize}
\begin{figure}
\begin{center}
\includegraphics[width=5.5in,height=3.5in]{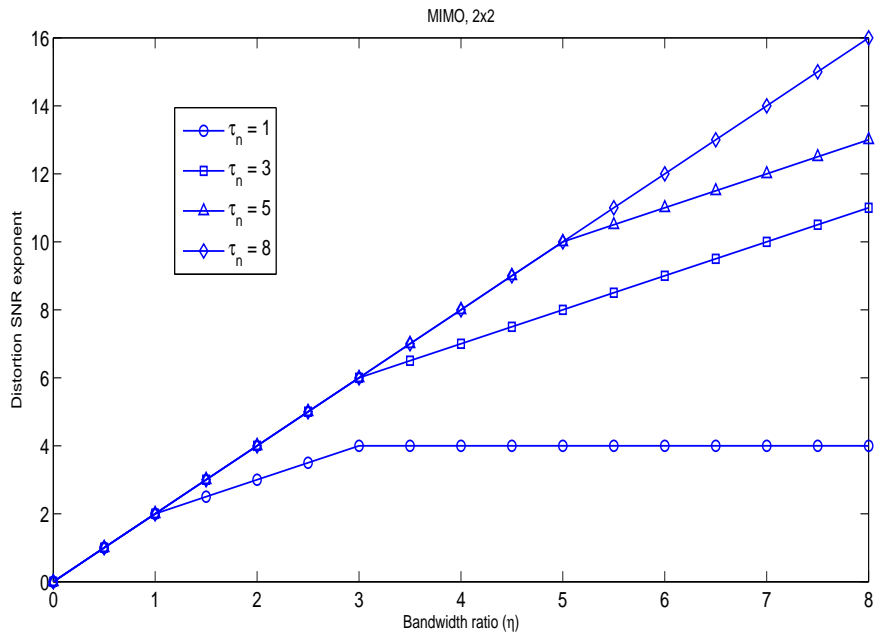}
\caption{Distortion exponent v.s. bandwidth ratio for block fading
2x2 MIMO channel. } \label{Fig:7}
\end{center}
\end{figure}

\begin{figure}
\begin{center}
\includegraphics[width=5.5in,height=3.5in]{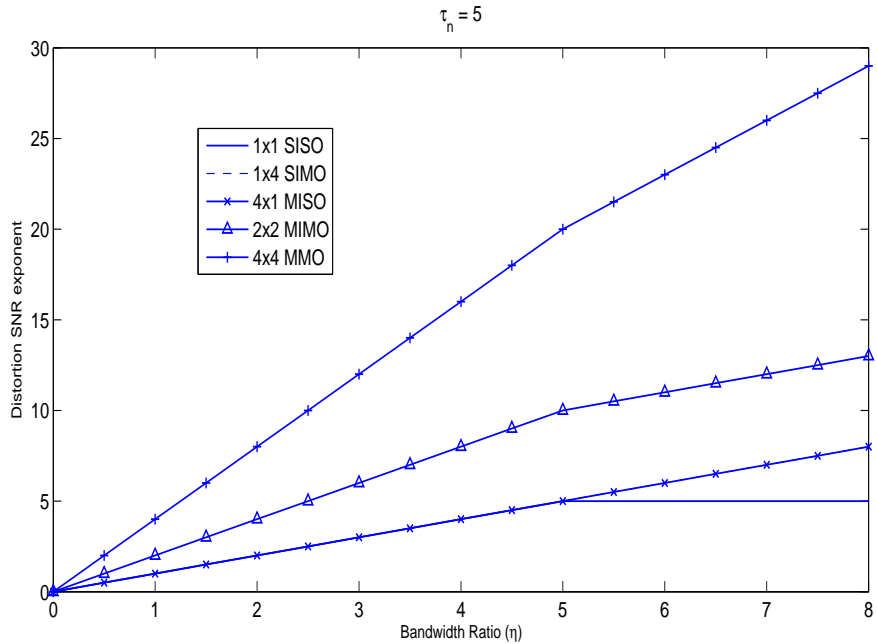}
\caption{Distortion exponent v.s. bandwidth ratio for normalized
delay = 5.} \label{Fig:8}
\end{center}
\end{figure}

In Fig. 7, we fixed the MIMO channel as $2\times2$, and plotted the
SNR exponent v.s. the bandwidth ratio curves for different delays.
As the delay increases, we have more time diversity to combat
fading, hence the corner point of  the exponent-bandwidth ratio
curve also increases. For $\tau_n = 1$, the maximum SNR exponent can
be achieved for $\eta = 3$. It is useless to increase channel
bandwidth ratio beyond 3 in the high SNR. In Fig. 8, We fixed the
normalized delay as $\tau_n = 5$ and show different SNR
exponent-bandwidth ratio curves for different antenna settings. For
SISO channel, the SNR exponent will not increase anymore as the
bandwidth ratio increase beyond $5$.

\subsection{MIMO Mutual Information Gaussian Approximation}
Due to inamenable to handle of Eqn. (\ref{eq:th1_1}), we can use
some approximations of the MIMO mutual information. The mathematical
operation of $\log\det(\cdot)$ involves an extensive amount of
average. Therefore the Lyapunov's central limit theorem can be
applied. The mutual information can be approximate as a Gaussian
distribution for large antenna systems. In \cite{HMT04}, the mean
and variance of different antenna settings has been derived. We will
use the results of \cite{HMT04} to derive the distortion-delay
approximations for different antenna settings.
\subsubsection{Large $M_r$, fixed $M_t$}
For this case the mutual information obeys
\begin{equation}\label{eq:Gau_1}
\mathcal{I}\sim \mathcal{N}\bigg( M_t \ln\Big(1+\frac{M_r\rho}{M_t}
\Big),\frac{M_t}{M_r}\bigg)~.
\end{equation}
The well-known moment generate function of the Gaussian distribution
is $\textrm{E}(e^{sx})=\exp(sm_x+\frac{1}{2}s^2\sigma_x^2)$, where
$m_x$ and $\sigma_x^2$ is the mean and variance of the Gaussian
variable $x$. Plug (\ref{eq:Gau_1}) into (\ref{eq:eff_cap}) and
after some straightforward math manipulations, we can get the
effective capacity and distortion delay function as
\begin{align}\label{eq:Gau_2.1}
E_c(\theta) &=
M_t\eta\ln\Big(1+\frac{M_r}{M_t}\rho\Big)-\frac{1}{2}\theta
K\frac{M_t}{M_r}\eta^2
\end{align}
\begin{align}\label{eq:Gau_2.2}
 \mathcal{D}(\tau_n) &\leq
\bigg[1+\frac{M_r\rho}{M_t}-\exp\Big(\frac{M_t}{2M_r}(\frac{\eta^2}{\tau_n})\Big)\bigg]^{-M_t\eta}
\end{align}
From Eqn. (\ref{eq:Gau_2.1}, \ref{eq:Gau_2.2}), the effective
capacity approaches to the ergodic capacity as $\theta \rightarrow
0$ or $M_r\rightarrow \infty$ (channel hardening). The SNR exponent
is $M_t\eta$, which is the same as Theorem 4, as $M_t$ fixed and
$M_r$ goes to infinity. Hence the SNR exponent is determined by the
first term in Eqn. (\ref{eq:exp_3}). We found the Guassian
approximation did reveal the distortion-delay tradeoff
asymptotically.
\subsubsection{Large $M_t$, fixed $M_r$}
the mutual information obeys
\begin{equation}\label{eq:Gau_3}
\mathcal{I}\sim \mathcal{N}\bigg( M_r \ln\Big(1+\rho\
\Big),\frac{M_r\rho^2}{M_t(1+\rho)^2}\bigg)~.
\end{equation}
The effective capacity and distortion delay curve is
\begin{align}\label{eq:Gau_4}
E_c(\theta) &= M_r\eta\ln(1+\rho)-\frac{1}{2}\theta
K\eta^2\frac{M_t}{M_r}\frac{\rho^2}{1+\rho^2}\\
\mathcal{D}(\tau_n) &\leq \bigg[1+\rho
-\exp\Big(\frac{M_r}{2M_t}(\frac{\eta^2}{\tau_n})\frac{\rho^2}{1+\rho^2}\Big)\bigg]^{-M_r\eta}
\end{align}
Again, the effective capacity approaches to the ergodic capacity as
$\theta \rightarrow 0$ or $M_t\rightarrow \infty$  The SNR exponent
is $M_r\eta$, which confirmed the results of Theorem 4.

\subsubsection{Large $M_t$ and $M_r$, Fixed $\beta = M_r/M_t$, High
SNR}
 The mutual information obeys
\begin{align}\label{eq:Gau_5}
\mathcal{I}&\sim \mathcal{N}\bigg( M_t
\mu(\beta,\rho),\sigma^2(\beta)\bigg)~, \quad \beta\geq 1 \\
&\sim \mathcal{N}\bigg( M_r \mu\Big(\frac{1}{\beta},\beta
\rho\Big),\sigma^2\Big(\frac{1}{\beta}\Big)\bigg)~, \quad \beta\leq~
1~.
\end{align}
Where $\mu(\beta,\rho) = \ln\rho+F(\beta)$,
$F(\beta),\sigma^2(\beta)$ are functions only depends on $\beta$.
The effective capacity capacity and distortion-delay function is:
\begin{align}\label{eq:Gau_4}
E_c(\theta) &= M_r\eta\ln(\rho)-\theta C_1 \\
\mathcal{D}(\tau_n)&\leq
\bigg[\rho -C_2\bigg]^{-M_r\eta}~, \quad \beta\geq 1 \\
E_c(\theta) &= M_t\eta\ln(\rho)-\theta C_3 \\
\mathcal{D}(\tau_n) &\leq \bigg[\rho -C_4\bigg]^{-M_t\eta}~, \quad
\beta\geq 1~,
\end{align}
Where $C_1,C_2,C_3,C_4$ are some constants. As both $M_r,M_t$ goes
to big and with fixed $\beta$, hence the $|M_t-M_r|$ also goes
large, the SNR exponent is still $M_*\eta$.
\section{Discussion and Remarks}
In previous sections, we have clearly characterized the
distortion/delay curve. However, we depend on some ideal
assumptions, e.g., the instantaneous channel capacity is achievable
and the CSI is perfectly known at the transmitter.
\begin{notation}\label{rm:1}(\emph{Decoding Error Probability})
In previous discussion we have assume using the the Gaussian code to
achieve the instantaneous capacity. In reality, we have to take the
decoding error probability into account for short codewords.
\cite{NG04} has integrated the physical layer decoding error into
the effective capacity function through random coding error
exponent. They have shown a joint queuing/coding exponent exits.
Such an exponent can fit well into our distortion and delay
analytical frame work.
\end{notation}
\begin{notation}\label{rm:2}(\emph{Power Control})
Since we have perfect CSI at the transmitter, given an average
transmission power budget, we can control the transmission power to
maximize the effective capacity or minimize the end-to-end
distortion for some delay constraint. In other words, the
transmission power is not necessarily constant. Recent work
\cite{JZ05} shows that, the optimum power adaptation policy is
related to the delay constraint. As the delay goes to infinity, the
power control policy approaches water-filling solution. On the
contrary, for stringent delay constraints, the optimum power control
policy becomes more like ``truncated channel inversion''. In the
future work, we will investigate the how optimum power control
affects the distortion/delay curves.
\end{notation}
\begin{notation}\label{rm:3}(\emph{Channel Correlation})
Although i.i.d. block fading channel is easy to analyze and has
several practical applications, this model is not always valid. It
is more general and practical to consider channel correlation. We
can use Jake's model to characterize the correlated channel fading
process. The autocorrelation of channel gain $R(\tau)$ can be
expressed as
\begin{equation}\label{eq:18}
R(\tau)=J_0(2\pi f_d \tau)~,
\end{equation}
where $J_0(\cdot)$ denotes the zero-th order Bessel function of
first kind and $f_d$ represents the maximum Doppler frequency.
Channel correlation will reduce the effective capacity\cite{WN03}.
Intuitively, correlation may cause the fading channel to stay in the
bad status for a longer time compared with i.i.d. block fading.
\cite{JZ05} shows that given a correlated fading channel with the
same marginal statistics as i.i.d. case, the effective capacity of
such a correlated channel is a linear shift in delay axis in
logarithmic scale, the shift value is proportional to the Doppler
frequency $f_d$. Hence the i.i.d. block fading distortion/delay
tradeoff can be easily extended to the correlated case.
\end{notation}

\section{Conclusion}
We investigate the fundamental problem of distortion/delay tradeoff
for the analogue source transmitted over wireless fading channels.
We derive a close-form analytical formula to characterize this
relationship using recently proposed effective capacity. Based on
this closed-form expression, we give out an upper bound that is
asymptotically tight to study the convergence behavior of the
distortion/delay function for SISO channel. We also characterized
the SNR exponent of MIMO block fading channel in the high SNR
regime. Simulation results show that a small delay can result in a
significant transmission power save. The framework of this paper is
applicable to a broad class application, e.g., video transmission.

\renewcommand{\thesection}{Apendix \Alph{section}}
\renewcommand{\theequation}{A-\arabic{equation}}
\setcounter{equation}{0}  
\setcounter{section}{0}

\section{Proof of Theorem 2}
\begin{proof}
From Eqn. (\ref{eq:cor1_1}) of Corollary, we have
\begin{align}\label{eq:aeq:3.5}
\mathcal{D}(\lambda)&\leq \bigg[\rho
^{-\lambda}\exp\Big(\frac{1}{\rho}\Big)\gamma\Big(1-\lambda,\frac{1}{\rho}\Big)\bigg]^{\frac{1}{\lambda}}\notag\\
&=\bigg[\frac{1}{\lambda-1}\frac{1}{\rho}
 \,_1F_1\Big(1;2-\lambda;\frac{1}{\rho}\Big)
+\Gamma(1-\lambda)\Big(\frac{1}{\rho}\Big)^{\lambda}
\exp\Big(\frac{1}{\rho}\Big)\bigg] ^{\frac{1}{\lambda}}
\end{align}
Since $\frac{1}{\lambda-1}<0$ as $\lambda \rightarrow 0$, we first
lower-bound the confluent hypergeometric function.
\begin{align}\label{aeq:4}
_1F_1(1;2-\lambda;x)
&=\sum_{k=0}^{\infty}\frac{(1)_k}{(2-\lambda)_k}
\frac{x^k}{k!}\notag\\ &\geq
\sum_{k=0}^{\infty}\frac{(1)_k}{(2)_k}\frac{x^k}{k!}=\frac{1}{x}(e^x-1)~,
\end{align}
where $(a)_k\triangleq a\cdot(a+1)\cdots(a+k-1)$. For
$\lambda\rightarrow0$ this lower bound is asymptotically tight. Next
we upper-bound the $\Gamma(1-\lambda)$.
\begin{align}\label{aeq:5}
\Gamma(1-\lambda)&=-\lambda\cdot\Gamma(-\lambda)
=\frac{-\lambda}{\frac{1}{\Gamma(-\lambda)}}\notag\\
&=\frac{-\lambda}{-\lambda+\xi(-\lambda)^2+\phi(-\lambda)^3+\delta(-\lambda)^4
+O((-\lambda)^5)} \notag \\
&\leq \frac{1}{1-\xi\lambda+\phi\lambda^2-\delta\lambda^3}~,
\end{align}
where $\xi=0.577215$ , $\phi = \frac{1}{12}(6\xi^2-\pi^2)$ and
$\delta$ is some constant. Hence replacing (\ref{aeq:4}) and
(\ref{aeq:5}) in (\ref{eq:15}) we have the following upper bound
\begin{equation}\label{aeq:6}
\mathcal{D}(\lambda) \dot{\leq}\bigg[
\frac{1}{\lambda-1}\big(e^{\frac{1}{\rho}}-1\big)+
\frac{1}{1-\xi\lambda+\phi\lambda^2}\rho^{-\lambda}e^{\frac{1}{\rho}}\bigg]
^{\frac{1}{\lambda}}~,
\end{equation}
where we have omitted $O(\lambda^3)$ term, which will not affect the
result as $\lambda\rightarrow 0$. Using Taylor expansion for the
first term and second term, and dropping the $O(\lambda^3)$, we
obtain the following asymptotic approximation,
\begin{align}\label{aeq:7}
\mathcal{D}_{upper}(\lambda)&\dot{\approx}
[1+a\lambda+b\lambda^2]^{\frac{1}{\lambda }}\notag\\
&= \exp(a)\exp\Big((b-\frac{a^2}{2})\lambda\Big)~,
\end{align}
where we have used the identity
$\lim_{x\rightarrow0}(1+x)^{\frac{1}{x}} = e$, and
\begin{align} \notag
a&\triangleq 1-e^{\frac{1}{\rho}}+\xi
e^{\frac{1}{\rho}}-\ln \rho e^{\frac{1}{\rho}}\\
b&\triangleq 1-e^{\frac{1}{\rho}}+(\xi^2-\phi)e^{\frac{1}{\rho}}-\xi
\ln\rho e^{\frac{1}{\rho}}+\ln^2\rho \notag~.
\end{align}
In order to show $\mathcal{D}_{upper}(\lambda)\rightarrow
\mathcal{D}(\infty)$ in (\ref{eq:17}), in other word
$(\ref{aeq:7})\rightarrow (\ref{eq:17})$, we want to show that
\begin{equation}\label{aeq:8}
F\triangleq1-e^{-\frac{1}{\rho}}-\xi + \ln\rho\rightarrow
E_1(1/\rho)~.
\end{equation}
$E_1(\cdot)$ is a special function, and don't have simple
expression. Instead we use numerical method to illustrate the
convergence. We have plotted these two values in Figure 6. We can
observe for most SNR these two values match perfectly. Hence we
conclude that the upper bound converges and the convergent rate is
exponential.
\begin{figure}
\begin{center}
\includegraphics[width=5in,height=3.5in]{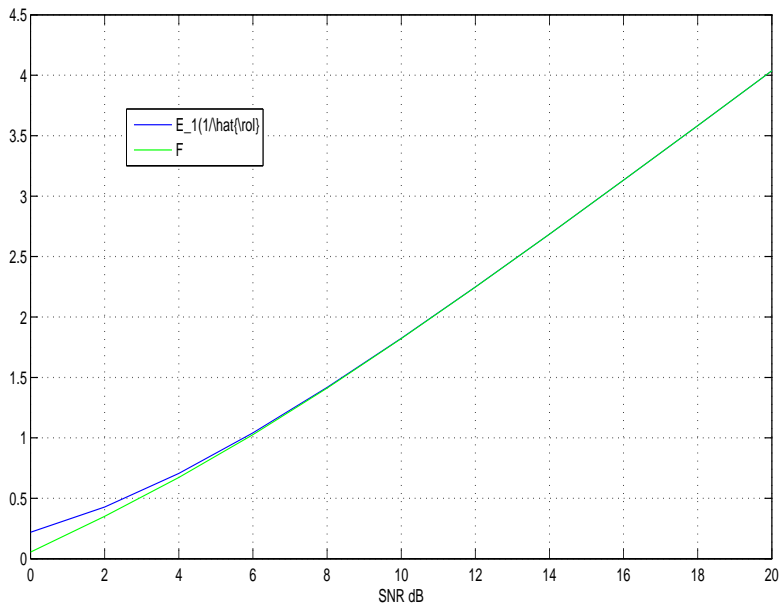}
\caption{Illustration (A-8) for different SNR values } \label{Fig:2}
\end{center}
\end{figure}
\end{proof}

\section{Proof of Theorem 4}
\begin{proof}
We will follow the technique used in \cite{ZT03}. Assume without
loss of generality that $M_t = M_* \leq M_r$ (the case  $M_t > M-r$
is a simple extension). We start from the distortion delay function
(\ref{eq:th1_4})
\begin{align}\label{aeq:9}
\mathcal{D}(\rho)=\bigg\{\int^{\infty}_{0}\prod
\Big(1+\frac{\rho}{M_t}\lambda_i\Big)^{-\theta K
\eta}f(\boldsymbol{\lambda})d\boldsymbol{\lambda}\bigg\}^{\frac{1}{\theta
K}}~,
\end{align}
where $\lambda_1 \leq \lambda_2 \leq \cdots \leq \lambda_{M_t}$ are
the ordered eigenvalues of $\mathbf{HH}^H$. We make the change of
variable: $\alpha_i= -\log(\lambda_i)/\log(\rho)$, for all $i =
1,\cdots,M_t$, The joint pdf $\boldsymbol{\alpha}=
[\alpha_1,\cdots,\alpha_{M_t}]$, where $\alpha_1\geq \cdots
\geq\alpha_{M_t}$, is given by
\begin{align}\label{aeq:10}
f(\boldsymbol{\alpha})=
K^{-1}_{M_t,M_r}\big(\log\rho\big)^{M_t}\prod^{M_t}_{i=1}\rho^{-(M_r-M_t+1)\alpha_i}
\prod_{i<j}\big(\rho^{-\alpha_i}-\rho^{-\alpha_j}\big)^2\exp\Big(\sum_i\rho^{-\alpha_i}\Big)~.
\end{align}
Replace $\boldsymbol{\lambda}$ with $\boldsymbol{\alpha}$,
(\ref{aeq:9}) yields
\begin{align}\label{aeq:11}
\mathcal{D}(\rho)=\bigg\{\int_{\mathcal{A}}\prod_{i=1}^{M_t}
(1+\frac{1}{M_t}\rho^{1-\alpha_i})^{-\theta
K\eta}f(\boldsymbol{\alpha})d\boldsymbol{\alpha}\bigg\}^{\frac{1}{\theta
K}}~,
\end{align}
where
\begin{equation}
\mathcal{A} = \Big\{\boldsymbol{\alpha}\in \mathbb{R}^{M_t}\quad :
\quad \alpha_1\geq \cdots \geq\alpha_{M_t}\Big\}\notag~.
\end{equation}
Neglecting all terms that irrelevant to the SNR exponent, we obtain
(\ref{aeq:9}) yields
\begin{align}\label{aeq:11}
\mathcal{D}(\rho)&\dot{\geq}\Bigg\{\int_{\mathcal{A}\bigcap\mathbb{R}^{M_t}{+}}\bigg(\prod_{i=1}^{M_t}
(1+\frac{1}{M_t}\rho^{1-\alpha_i})^{-\theta
K\eta}\bigg)\prod^{M_t}_{i=1}\rho^{-(2i-1+M_r-M_t)\alpha_i}d\boldsymbol{\alpha}\Bigg\}^{\frac{1}{\theta
K}} \notag\\
&\dot{=}\bigg\{\int_{\mathcal{A}\bigcap\mathbb{R}^{M_t}{+}}
\prod_{i=1}^{M_t}\rho^{-\theta K\eta(1-\alpha_i)^+}
\prod_{i=1}^{M_t}\rho^{-(2i-1+M_r-M_t)\alpha_i}d\boldsymbol{\alpha}\bigg\}^{\frac{1}{\theta
K}}\notag\\
 &\dot{=}\bigg\{\int_{\mathcal{A}\bigcap\mathbb{R}^{M_t}{+}} \prod_{i=1}^{M_t}\rho^{-(\theta K\eta(1-\alpha_i)^+
+(2i-1+M_r-M_t)\alpha_i)} \bigg\}^{\frac{1}{\theta K}}\notag\\
 &\dot{=}\rho^{\alpha(\eta)\frac{1}{\theta K}}
\end{align}
where we have used
\begin{equation}
(1+\frac{1}{M_t}\rho^{1-\alpha_i})^{-\theta K\eta}
\dot{=}\rho^{-\theta K \eta[1-\alpha_i]^+} \notag~.
\end{equation}
And
\begin{equation}
\alpha(\eta) = \inf_{\boldsymbol{\alpha}\in
\mathcal{A}\bigcap\mathbb{R}^{M_t}{+}}\sum_{i=1}^{M_r}(2i-1+M_r-M-t)\alpha_i+\theta
K\eta(1-\alpha_i)^+\notag ~.
\end{equation}
We can minimizing individual term of the summation separately by set
$\alpha_i= 0$ or $1$. We also notice that $\theta K = \tau_n$, the
buffer delay, hence we can obtain the SNR exponent of the buffered
transmission is
\begin{equation}
\alpha(\eta) =\tau_n \min
\Big\{\frac{\eta}{\tau_n},2i-1+M_r-M_t\Big\}~.
\end{equation}
\end{proof}


\end{document}